\documentclass[a4paper,12pt,reqno]{amsart}

\usepackage{graphicx,amssymb,datetime,float,MnSymbol,currfile,tikz}
\usepackage[round]{natbib}
\usetikzlibrary{arrows}
\usetikzlibrary{decorations.markings}

\newtheorem{theorem}{Theorem}

\newtheorem{corollary}[theorem]{Corollary}

\newtheorem{example}[theorem]{Example}

\def\ra{\rightarrow}

\newcommand{\comments}[1]{}

\def\na{\overline{a}}
\def\nc{\overline{c}}
\def\nd{\overline{d}}

\tikzset{tt/.style={decoration={
  markings,
  mark=at position .485 with {\arrow{>}},
  mark=at position .515 with {\arrow{<}}},postaction={decorate}}}

\begin{document}

\title[]{On the Monotonicity of a Nondifferentially Mismeasured Binary Confounder}

\author[]{Jose M. Pe\~{n}a\\
IDA, Link\"oping University, Sweden\\
jose.m.pena@liu.se}

\date{\currenttime, \ddmmyydate{\today}, \currfilename}

\maketitle

\begin{abstract}
Suppose that we are interested in the average causal effect of a binary treatment on an outcome when this relationship is confounded by a binary confounder. Suppose that the confounder is unobserved but a nondifferential proxy of it is observed. We show that, under certain monotonicity assumption that is empirically verifiable, adjusting for the proxy produces a measure of the effect that is between the unadjusted and the true measures.
\end{abstract}

\begin{figure}[t]
\begin{tabular}{c|c}
\begin{tikzpicture}[inner sep=1mm]
\node at (0,0) (A) {$A$};
\node at (1,.5) (D) {$D$};
\node at (2,0) (Y) {$Y$};
\node at (1,1.5) (C) {$C$};
\path[->] (A) edge (Y);
\path[->] (C) edge (A);
\path[->] (C) edge (Y);
\path[->] (C) edge (D);
\end{tikzpicture}
&
\begin{tikzpicture}[inner sep=1mm]
\node at (0,0) (A) {$A$};
\node at (1,.5) (D) {$D$};
\node at (2,0) (Y) {$Y$};
\node at (1,1.5) (C) {$C$};
\path[->] (A) edge (Y);
\path[->] (C) edge (A);
\path[->] (C) edge (Y);
\path[<-] (C) edge (D);
\end{tikzpicture}
\end{tabular}\caption{Causal graphs, where $Y$ is a discrete or continuous random variable, and $A$, $C$ and $D$ are binary random variables. Moreover, $C$ is unobserved.}\label{fig:graphs}
\end{figure}
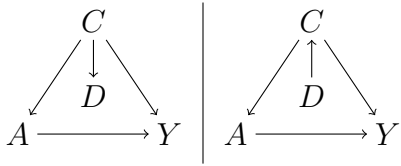

\section{Introduction}

Suppose that we are interested in the average causal effect of a binary treatment $A$ on an outcome $Y$ when this relationship is confounded by a binary confounder $C$. Suppose also that $C$ is nondifferentially mismeasured, meaning that (i) $C$ is not observed and, instead, a binary proxy $D$ of $C$ is observed, and (ii) $D$ is conditionally independent of $A$ and $Y$ given $C$. The causal graph to the left in Figure \ref{fig:graphs} represents the relationships between the random variables.

\citet{Greenland1980} argues that adjusting for $D$ produces a partially adjusted measure of the average causal effect of $A$ on $Y$ that is between the crude (i.e., unadjusted) and true (i.e., adjusted for $C$) measures. \citet[Lemma 1]{OgburnandVanderWeele2012a} show that, although this result does not always hold, it does hold under some monotonicity condition in $C$. Specifically, $E[Y|A,C]$ must be nondecreasing or nonincreasing in $C$. Since this condition can be interpreted as that the average causal effect of $C$ on $Y$ must be in the same direction among the treated ($A=1$) and the untreated ($A=0$), \cite{OgburnandVanderWeele2012a} argue that the condition is likely to hold in most applications in epidemiology. Unfortunately, the condition cannot be verified empirically because $C$ is unobserved. Therefore, one has to rely on substantive knowledge to verify it. Moreover, the condition is sufficient but not necessary. \cite{OgburnandVanderWeele2013} extend these results to the case where $C$ takes more than two values. If there are at least two independent proxies of $C$, then \citet{Miaoetal.2018} show that the causal effect of $A$ on $Y$ can be identified under certain rank condition.

In this paper, we prove that if the monotonicity condition holds in $D$, then it holds in $C$ as well. Since $D$ is observed, the monotonicity condition in $D$ can be verified empirically. Therefore, if no substantive knowledge is available but data are, then combining our result with Lemma 1 by \citet{OgburnandVanderWeele2012a} may allow us to conclude that the partially adjusted effect is between the crude and the true ones and, thus, that the partially adjusted effect is a better approximation to the true effect than the crude one. We also report experiments showing that most random parameterizations of the causal graph to the left in Figure \ref{fig:graphs} result in a partially adjusted effect that lies between the crude and the true ones, although only half of them satisfy the monotonicity condition in $D$. This confirms that the condition is sufficient but not necessary. This result should be interpreted with caution because, in fields like epidemiology, one is not typically concerned with a random parameterization but, rather, with one carefully engineered by evolution. We provide a partial answer to this question by characterizing a nonmonotonic case (albeit empirically untestable) where the partially adjusted effect still lies between the crude and the true ones. Finally, we also prove that if the monotonicity condition holds in $D$, then it also holds in $C$ when $D$ is a driver of $C$ rather than a proxy, i.e. $D$ causes $C$. We illustrate the relevance of this result with an example on transportability of causal inference across populations.

The rest of the paper is organized as follows. Sections \ref{sec:proxy} and \ref{sec:driver} present our results when $D$ is a proxy and a driver of $C$, respectively. Section \ref{sec:discussion} closes with some discussion.

\section{On a Proxy of the Confounder}\label{sec:proxy}

Consider the causal graph to the left in Figure \ref{fig:graphs}, where $Y$ is a discrete or continuous random variable, and $A$, $C$ and $D$ are binary random variables. The graph entails the following factorization:
\begin{equation}\label{eq:factorization}
p(A,C,D,Y)=p(C)p(D|C)p(A|C)p(Y|A,C).
\end{equation}
Let $A$ take values $a$ and $\na$, and similarly for $C$ and $D$. Let $A$, $D$ and $Y$ be observed and let $C$ be unobserved. Let $Y_a$ and $Y_{\na}$ denote the counterfactual outcomes under treatments $A=a$ and $A=\na$, respectively. The average causal effect of $A$ on $Y$ or true risk difference ($RD_{true}$) is defined as $RD_{true}=E[Y_a]-E[Y_{\na}]$. It can be rewritten as follows \cite[Theorem 3.3.2]{Pearl2009}:
\[
RD_{true} = E[Y|a,c]p(c) + E[Y|a,\nc]p(\nc) - E[Y|\na,c]p(c) - E[Y|\na,\nc]p(\nc).
\]
Since $C$ is unobserved, $RD_{true}$ cannot be computed. It can be approximated by the unadjusted average causal effect or crude risk difference ($RD_{crude}$):
\[
RD_{crude} = E[Y|a] - E[Y|\na]
\]
and by the partially adjusted average causal effect or observed risk difference ($RD_{obs}$):
\[
RD_{obs} = E[Y|a,d]p(d) + E[Y|a,\nd]p(\nd) - E[Y|\na,d]p(d) - E[Y|\na,\nd]p(\nd).
\]

We say that $E[Y|A,D]$ is nondecreasing in $D$ if 
\begin{equation}\label{eq:nd}
E[Y|a, d] \geq E[Y|a, \nd] \text{ and } E[Y|\na, d] \geq E[Y|\na, \nd].
\end{equation}
Likewise, $E[Y|A,D]$ is nonincreasing in $D$ if
\begin{equation}\label{eq:ni}
E[Y|a, d] \leq E[Y|a, \nd] \text{ and } E[Y|\na, d] \leq E[Y|\na, \nd].
\end{equation}
Moreover, $E[Y|A,D]$ is monotone in $D$ if it is nondecreasing or nonincreasing in $D$. \citet[Lemma 1]{OgburnandVanderWeele2012a} show that if $E[Y|A,C]$ is monotone in $C$, then $E[Y|A,D]$ is monotone in $D$. The following theorem proves the converse result. The relevance of this result is as follows. \citet[Result 1]{OgburnandVanderWeele2012a} show that if $E[Y|A,C]$ is monotone in $C$, then $RD_{obs}$ lies between $RD_{true}$ and $RD_{crude}$. The antecedent of this rule cannot be verified empirically, because $C$ is unobserved. Therefore, one must rely on substantive knowledge to apply the rule. The following theorem implies that, luckily, the rule also holds for $D$ and, thus, that the antecedent can be verified empirically.

\begin{theorem}\label{the:theorem}
Consider the causal graph to the left in Figure \ref{fig:graphs}. If $E[Y|A,D]$ is monotone in $D$, then $E[Y|A,C]$ is monotone in $C$.
\end{theorem}

\begin{proof}
Assume to the contrary that $E[Y|A,C]$ is not monotone in $C$, i.e.
\begin{equation}\label{eq:nm1}
E[Y|a, c] \leq E[Y|a, \nc] \text{ and } E[Y|\na, c] \geq E[Y|\na, \nc]
\end{equation}
or
\begin{equation}\label{eq:nm2}
E[Y|a, c] \geq E[Y|a, \nc] \text{ and } E[Y|\na, c] \leq E[Y|\na, \nc].
\end{equation}
This gives four cases to consider: Whether Equation \ref{eq:nd} or \ref{eq:ni} holds, and whether Equation \ref{eq:nm1} or \ref{eq:nm2} holds. Hereinafter, we focus on the first case. The other cases are similar.

Assume that Equations \ref{eq:nd} and \ref{eq:nm1} hold. We show next that the first inequalities in Equations \ref{eq:nd} and \ref{eq:nm1} imply that $p(c|a,d) \leq p(c|a, \nd)$. Specifically, 
\begin{align*}
E[Y|a, d] &\geq E[Y|a, \nd]\\
E[Y|a, d, c] p(c|a,d) + E[Y|a, d, \nc] p(\nc|a,d) &\geq\\
E[Y|a, \nd, c] p(c|a, \nd) + E[Y|a, \nd, \nc] p(\nc|a, \nd) &\\
E[Y|a, c] p(c|a,d) + E[Y|a, \nc] p(\nc|a,d) &\geq\\
E[Y|a, c] p(c|a, \nd) + E[Y|a, \nc] p(\nc|a, \nd)&
\end{align*}
because $Y$ is conditionally independent of $D$ given $A$ and $C$ due to the causal graph under consideration and, thus,
\begin{align*}
E[Y|a, c] p(c|a,d) + E[Y|a, \nc] (1 - p(c|a,d)) &\geq\\
E[Y|a, c] p(c|a, \nd) + E[Y|a, \nc] (1 - p(c|a, \nd)) &\\
(E[Y|a, c] - E[Y|a, \nc]) p(c|a,d) &\geq\\
(E[Y|a, c] - E[Y|a, \nc]) p(c|a, \nd) &\\
p(c|a,d) &\leq p(c|a, \nd)
\end{align*}
because $E[Y|a, c] \leq E[Y|a, \nc]$ by Equation \ref{eq:nm1}.

Furthermore,
\[
p(c|a,d) = \frac{p(a,d|c)p(c)}{p(a,d|c)p(c) + p(a,d|\nc)p(\nc)} = \frac{1}{1+\exp(-\delta(a,d))} = \sigma(\delta(a,d))
\]
where
\[
\delta(a,d) = \ln \frac{p(a,d|c)p(c)}{p(a,d|\nc)p(\nc)}
\]
is known as the log odds, and $\sigma()$ is known as the logistic sigmoid function \cite[Section 4.2]{Bishop2006}. Note that $\sigma()$ is an increasing function. Then,
\begin{align*}
p(c|a,d) &\leq p(c|a, \nd)\\
\delta(a,d) &\leq \delta(a,\nd)\\
\ln p(a|c) + \ln p(d|c) + \ln p(c) &-\\
 \ln p(a|\nc) - \ln p(d|\nc) - \ln p(\nc) &\leq\\
\ln p(a|c) + \ln p(\nd|c) + \ln p(c) &-\\
 \ln p(a|\nc) - \ln p(\nd|\nc) - \ln p(\nc)&
\end{align*}
because $A$ is conditionally independent of $D$ given $C$ due to the causal graph under consideration and, thus,
\begin{align}\nonumber
\ln p(d|c) - \ln p(d|\nc) &\leq \ln p(\nd|c) - \ln p(\nd|\nc)\\\nonumber
\ln \frac{p(d|c)}{p(d|\nc)} &\leq \ln \frac{p(\nd|c)}{p(\nd|\nc)}\\
\frac{p(d|c)}{p(d|\nc)} &\leq \frac{p(\nd|c)}{p(\nd|\nc)}.\label{eq:leq}
\end{align}

Likewise, the second inequalities in Equations \ref{eq:nd} and \ref{eq:nm1} imply that $p(c|\na,d) \geq p(c|\na, \nd)$, which implies that
\[
\frac{p(d|c)}{p(d|\nc)} \geq \frac{p(\nd|c)}{p(\nd|\nc)}
\]
which contradicts Equation \ref{eq:leq} unless equality holds. However, equality only occurs if $p(d|c)=p(d|\nc)$, which implies that $C$ and $D$ are independent and, thus, that $D$ is not a mismeasured confounder.
\end{proof}

\begin{corollary}\label{cor:corollary}
Consider the causal graph to the left in Figure \ref{fig:graphs}. If $E[Y|A,D]$ is monotone in $D$, then $RD_{obs}$ lies between $RD_{true}$ and $RD_{crude}$.
\end{corollary}

\begin{proof}
The result follows directly from Theorem \ref{the:theorem} and \citet[Result 1]{OgburnandVanderWeele2012a}.
\end{proof}

\begin{table}[t]
\caption{Results of 10000 random parameterizations of the causal graph to the left in Figure \ref{fig:graphs}.}\label{tab:results}
\begin{tabular}{|r|l|r|r|r|}
\hline
In-between && Nondec. in $D$ & Noninc. in $D$ & Neither\\
\hline
2430& Nondec. in $C$&1175&1255&0\\
\hline
2461& Noninc. in $C$&1225&1236&0\\
\hline
4460& Neither&0&0&5109\\
\hline
\end{tabular}
\end{table}

\begin{figure}[t]
\begin{tabular}{cc}
\includegraphics[scale=.5]{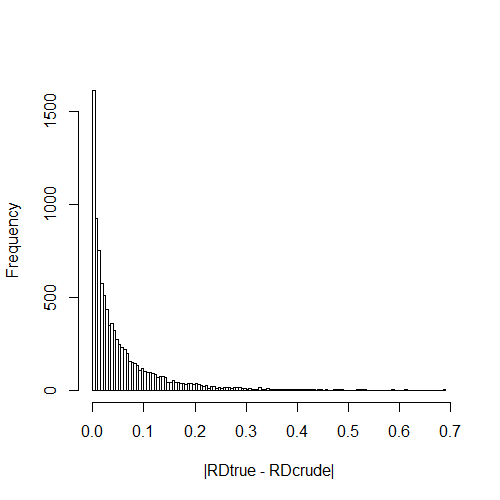}&\includegraphics[scale=.5]{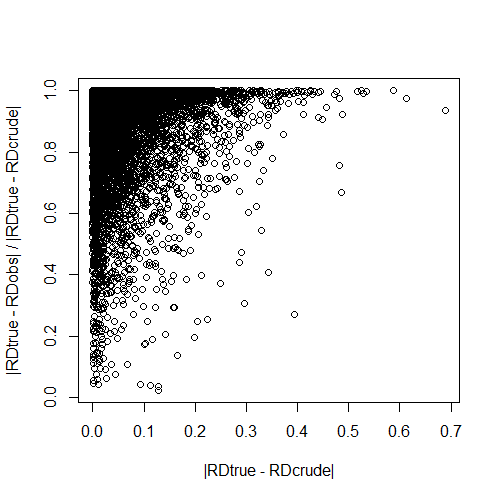}\\
\includegraphics[scale=.5]{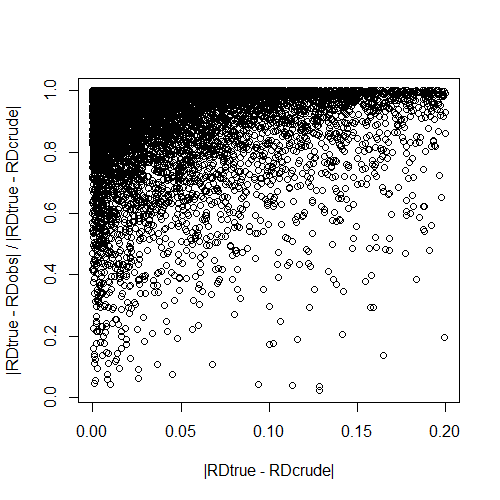}&\includegraphics[scale=.5]{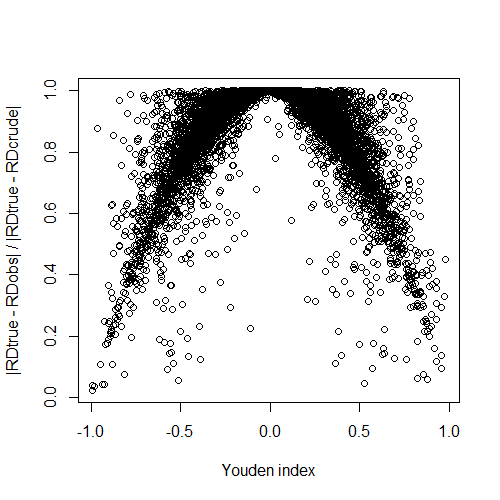}
\end{tabular}\caption{(tl) Histogram of interval length. (tr) Distance between $RD_{obs}$ and $RD_{true}$ relative to interval length. (bl) Zoom of previous plot. (br) Distance between $RD_{obs}$ and $RD_{true}$ relative to interval length, as a function of correlation between $C$ and $D$ when measured by Youden index.}\label{fig:plots}
\end{figure}

\subsection{Experiments}

In this section, we report some experiments that shed additional light on the relationships between the various risk differences. Specifically, we randomly parameterized 10000 times the causal graph to the left in Figure \ref{fig:graphs} by parameterizing the terms in the right-hand side of Equation \ref{eq:factorization} with parameter values drawn from a uniform distribution.\footnote{Code available at \texttt{https://www.dropbox.com/s/pa8y2sausib6hcn/monotonicity.R?dl=0}.} For each parameterization, we then computed $RD_{true}$, $RD_{obs}$ and $RD_{crude}$. The results are reported in Table \ref{tab:results}. Of the 10000 runs, 4891 were monotone in $C$ and also in $D$, as expected from \citet[Lemma 1]{OgburnandVanderWeele2012a}. There were no other runs that were monotone in $D$, as expected from Theorem \ref{the:theorem}. In all these 4891 runs, $RD_{obs}$ was between $RD_{true}$ and $RD_{crude}$, as expected from Corollary \ref{cor:corollary} and \citet[Result 1]{OgburnandVanderWeele2012a}. It is also worth noticing from the table that the 10000 runs are rather evenly distributed among the different entries. Finally, 4460 of the 5109 runs where the monotonicity assumption did not hold still resulted in that $RD_{obs}$ was between $RD_{true}$ and $RD_{crude}$. In other words, although half of the runs violated the monotonicity assumption, few of them resulted in $RD_{obs}$ being outside the range of $RD_{true}$ and $RD_{crude}$. In total, $RD_{obs}$ was between $RD_{true}$ and $RD_{crude}$ in 94 \% of the runs. Therefore, $RD_{obs}$ was a better approximation to $RD_{true}$ than $RD_{crude}$ in most of the runs. We investigate further this question in the next section, where we characterize a nonmonotonic case where $RD_{obs}$ still lies between $RD_{true}$ and $RD_{crude}$.

The plots in Figure \ref{fig:plots} show some additional descriptive statistics for the runs where $RD_{obs}$ belonged to the interval between $RD_{true}$ and $RD_{crude}$. The top left plot shows that most intervals were quite small and, thus, that $RD_{obs}$ was a good approximation to $RD_{true}$ in most cases. However, the top right plot shows that $RD_{obs}$ was typically closer to $RD_{crude}$ than to $RD_{true}$. The bottom left plot is a zoom of the previous plot at the smallest intervals. Finally, the bottom right plot shows that the lower the correlation between $C$ and $D$ when measured by the Youden index, the closer $RD_{obs}$ was to $RD_{crude}$. In summary, $RD_{obs}$ seems to be a good approximation to $RD_{true}$, but it seems to be biased towards $RD_{crude}$. This is a problem when the interval between $RD_{crude}$ and $RD_{true}$ is large. However, the length of the interval is unknown in practice, and we doubt substantive knowledge may provide hints on it. The bias seems to decrease with increasing correlation between $C$ and $D$. Although this correlation is unknown in practice, substantive knowledge may give hints on it.

\subsection{Nonmonotonicity}

Consider the causal graph to the left in Figure \ref{fig:graphs}. This section characterizes a case where $E[Y|A,C]$ is not monotone in $C$ and, thus, $E[Y|A,D]$ is not monotone in $D$ by Theorem \ref{the:theorem}, and yet $RD_{obs}$ lies between $RD_{true}$ and $RD_{crude}$. Specifically, let $A$, $D$ and $Y$ represent three diseases, and $C$ a risk factor for the three of them. Suppose that suffering $A$ affects the risk of suffering $Y$.  Suppose that half of the population is exposed to the risk factor $C$, i.e. $p(c)=0.5$. Suppose also that the exposure to $C$ affects the risk of suffering $A$ and $D$ as $p(a|c)=p(\na|\nc)=p(d|c)=p(\nd|\nc) \geq 0.5$. Finally, suppose that $E[Y|a,c] - E[Y|a,\nc] \geq 0$ and $E[Y|\na,\nc] - E[Y|\na,c] \geq 0$. In other words, the exposure to $C$ increases the average severity of the disease $Y$ for the individuals suffering the disease $A$, while it decreases the severity for the rest. Therefore, the monotonicity assumption does hold. However, under the additional assumption that $E[Y|a,c] - E[Y|a,\nc] \geq E[Y|\na,\nc] - E[Y|\na,c]$, we can still conclude that $RD_{obs}$ lies between $RD_{true}$ and $RD_{crude}$. Note that one has to rely on substantive knowledge to verify the conditions in the characterization, because $C$ is unobserved. The following theorems formalize this result.

\begin{theorem}\label{the:nonmonotone}
Consider the causal graph to the left in Figure \ref{fig:graphs}. Let $p(c)=0.5$ and $p(a|c)=p(\na|\nc)=p(d|c)=p(\nd|\nc) \geq 0.5$. If $E[Y|a,c] - E[Y|a,\nc] \geq E[Y|\na,\nc] - E[Y|\na,c] \geq 0$, then $RD_{crude} \geq RD_{obs} \geq RD_{true}$.
\end{theorem}

\begin{proof}
We start by proving some auxiliary facts.

{\bf Fact 1}. Recall from the proof of Theorem \ref{the:theorem} that
\[
p(c|a,d) = \sigma \Big( \ln \frac{p(a,d|c)p(c)}{p(a,d|\nc)p(\nc)} \Big) = \sigma \Big( \ln \frac{p(a|c)p(d|c)}{p(a|\nc)p(d|\nc)} \Big)
\]
where the last equality follows from the assumption that $p(c)=0.5$, and the fact that $A$ and $D$ are conditionally independent given $C$ due to the causal graph under consideration. Note that the previous equation implies that $p(c|a,d)=p(\nc|\na,\nd) \geq 0.5$ and $p(c|\na,d)=p(c|a,\nd)=0.5$, due to the assumption that $p(a|c)=p(\na|\nc)=p(d|c)=p(\nd|\nc) \geq 0.5$.

{\bf Fact 2}. Note that
\begin{align*}
E[Y|a,d] - E[Y|a,\nd] &= E[Y|a,c,d] p(c|a,d) + E[Y|a,\nc,d] p(\nc|a,d)\\
&- E[Y|a,c,\nd] p(c|a,\nd) - E[Y|a,\nc,\nd] p(\nc|a,\nd)\\
&= E[Y|a,c] (p(c|a,d) - 0.5) + E[Y|a,\nc] (p(\nc|a,d) - 0.5)\\
&= E[Y|a,c] (p(c|a,d) - 0.5) - E[Y|a,\nc] (p(c|a,d) - 0.5)
\end{align*}
where the second equality follows from the fact that $Y$ and $D$ are conditionally independent given $A$ and $C$ due to the causal graph under consideration, and the fact that $p(c|\na,d)=p(c|a,\nd)=0.5$ by Fact 1. Likewise,
\begin{align*}
E[Y|\na,\nd] - E[Y|\na,d] &= E[Y|\na,c,\nd] p(c|\na,\nd) + E[Y|\na,\nc,\nd] p(\nc|\na,\nd)\\
&- E[Y|\na,c,d] p(c|\na,d) - E[Y|\na,\nc,d] p(\nc|\na,d)\\
&= E[Y|\na,c] (p(c|\na,\nd) - 0.5) + E[Y|\na,\nc] (p(\nc|\na,\nd) - 0.5)\\
&= - E[Y|\na,c] (p(c|a,d) - 0.5) + E[Y|\na,\nc] (p(c|a,d) - 0.5)
\end{align*}
since $p(c|a,d)=p(\nc|\na,\nd)$ by Fact 1. Then, the assumption that $E[Y|a,c] - E[Y|a,\nc] \geq E[Y|\na,\nc] - E[Y|\na,c] \geq 0$ together with the fact that $p(c|a,d) \geq 0.5$ by Fact 1 imply that $E[Y|a,d] - E[Y|a,\nd] \geq E[Y|\na,\nd] - E[Y|\na,d] \geq 0$.

{\bf Fact 3}. Note that
\[
p(d)=p(d|c)p(c)+p(d|\nc)p(\nc)=p(\nd|\nc)p(\nc)+p(\nd|c)p(c)=p(\nd)
\]
by the assumptions that $p(c)=0.5$ and $p(d|c)=p(\nd|\nc)$. Then, $p(d)=0.5$ and, thus, $p(c|d)=p(d|c)=p(\nd|\nc)=p(\nc|\nd)$. Likewise, $p(a)=0.5$ and $p(d|a)=p(a|d)=p(\na|\nd)=p(\nd|\na)$. Note also that
\[
p(\na|d)=p(\na|c,d)p(c|d)+p(\na|\nc,d)p(\nc|d)=p(\na|c)p(c|d)+p(a|c)p(\nc|d)
\]
since $A$ and $D$ are conditionally independent given $C$ due to the causal graph under consideration, and $p(a|c)=p(\na|\nc)$ by assumption. Moreover, the previous equation can be rewritten as
\[
p(\na|d)=2 p(a|c) ( 1- p(a|c) )
\]
because $p(d|c)=p(c|d)$ as shown above, whereas $p(a|c) = p(d|c)$ by assumption. The last equation implies that $p(\na|d) \ngtr 0.5$. To see it, rewrite the last equation as the function $f(x)=2x(1-x)$. By inspecting the first and second derivatives, we can conclude that $f(x)$ has a single maximum at $x=0.5$ with value 0.5. That $p(\na|d) \ngtr 0.5$ implies that $p(d|a) = p(a|d) \geq 0.5$.

We now prove the theorem. Note that the assumption that $p(c)=0.5$ implies that
\begin{equation}\label{eq:rdtrue}
RD_{true} = ( E[Y|a,c] + E[Y|a,\nc] - E[Y|\na,c] - E[Y|\na,\nc] )/2.
\end{equation}
Note also that $p(d)=0.5$ by Fact 3 and, thus,
\begin{align}\label{eq:rdobs}
RD_{obs} &= ( E[Y|a,d] + E[Y|a,\nd] - E[Y|\na,d] - E[Y|\na,\nd] )/2\\\nonumber
&= \big( E[Y|a,c,d]p(c|a,d) + E[Y|a,\nc,d]p(\nc|a,d)\\\nonumber
&+ E[Y|a,c,\nd]p(c|a,\nd) + E[Y|a,\nc,\nd]p(\nc|a,\nd)\\\nonumber
&- E[Y|\na,c,d]p(c|\na,d) - E[Y|\na,\nc,d]p(\nc|\na,d)\\\nonumber
&- E[Y|\na,c,\nd]p(c|\na,\nd) - E[Y|\na,\nc,\nd]p(\nc|\na,\nd) \big)/2\\\nonumber
& = \big( E[Y|a,c] (p(c|a,d) + p(c|a,\nd))\\\nonumber
& + E[Y|a,\nc] (p(\nc|a,d) + p(\nc|a,\nd))\\\nonumber
& - E[Y|\na,c] (p(c|\na,d) + p(c|\na,\nd))\\\nonumber
& - E[Y|\na,\nc] (p(\nc|\na,d) + p(\nc|\na,\nd)) \big)/2\nonumber
\end{align}
since $Y$ and $D$ are conditionally independent given $A$ and $C$ due to the causal graph under consideration. Note that $p(c|a,d)=p(\nc|\na,\nd)$ and $p(c|\na,d)=p(c|a,\nd)=0.5$ by Fact 1. Then, the previous equation can be rewritten as follows with $\alpha = p(c|a,d) - 0.5$:
\begin{align}\nonumber
RD_{obs} &= \big( E[Y|a,c] (p(c|a,d) + 0.5) + E[Y|a,\nc] (p(\nc|a,d) + 0.5)\\\nonumber
& - E[Y|\na,c] (0.5 + p(\nc|a,d)) - E[Y|\na,\nc] (0.5 + p(c|a,d)) \big)/2\\\nonumber
&= \big( E[Y|a,c] (1 + \alpha) + E[Y|a,\nc] (1 - \alpha)\\\nonumber
& - E[Y|\na,c] (1 - \alpha) - E[Y|\na,\nc] (1 + \alpha) \big)/2\\\nonumber
&= ( E[Y|a,c] + E[Y|a,\nc] - E[Y|\na,c] - E[Y|\na,\nc] ) / 2\\\nonumber
& + ( E[Y|a,c] \alpha - E[Y|a,\nc] \alpha + E[Y|\na,c] \alpha - E[Y|\na,\nc] \alpha )/2\\\nonumber
&\geq RD_{true}\nonumber
\end{align}
by Equation \ref{eq:rdtrue} and the fact that the term in penultimate line above is nonnegative. The latter follows from the assumption that $E[Y|a,c] - E[Y|a,\nc] \geq E[Y|\na,\nc] - E[Y|\na,c] \geq 0$, and the fact that $\alpha = p(c|a,d) - 0.5 \geq 0$ by Fact 1.

Having proven that $RD_{obs} \geq RD_{true}$, it only remains to prove that $RD_{crude} \geq RD_{obs}$. Let $\beta=p(d|a) - 0.5$. Note that $p(d|a)=p(\nd|\na)$ by Fact 3. Then,
\begin{align}\nonumber
RD_{crude} &= E[Y|a] - E[Y|\na]\\\nonumber
&= E[Y|a,d] p(d|a) + E[Y|a,\nd] p(\nd|a)\\\nonumber
&- E[Y|\na,d] p(d|\na) - E[Y|\na,\nd] p(\nd|\na)\\\nonumber
&= E[Y|a,d] (0.5 + \beta) + E[Y|a,\nd] (0.5 - \beta)\\\nonumber
&- E[Y|\na,d] (0.5 - \beta) - E[Y|\na,\nd] (0.5 + \beta)\\\nonumber
&= E[Y|a,d] 0.5 + E[Y|a,\nd] 0.5 - E[Y|\na,d] 0.5 - E[Y|\na,\nd] 0.5\\\nonumber
&+ E[Y|a,d] \beta - E[Y|a,\nd] \beta + E[Y|\na,d] \beta - E[Y|\na,\nd] \beta\\\nonumber
&\geq RD_{obs}\nonumber
\end{align}
by Equation \ref{eq:rdobs} and the fact that the term in the penultimate line above is nonnegative. The latter follows from the fact that $E[Y|a,d] - E[Y|a,\nd] \geq E[Y|\na,\nd] - E[Y|\na,d] \geq 0$ by Fact 2, and the fact that $\beta = p(d|a) - 0.5 \geq 0$ by Fact 3.
\end{proof}

\begin{theorem}
Consider the causal graph to the left in Figure \ref{fig:graphs}. Let $p(c)=0.5$ and $p(a|c)=p(\na|\nc)=p(d|c)=p(\nd|\nc) \geq 0.5$. If $E[Y|a,c] - E[Y|a,\nc] \leq E[Y|\na,\nc] - E[Y|\na,c] \leq 0$, then $RD_{crude} \leq RD_{obs} \leq RD_{true}$.
\end{theorem}

\begin{proof}
The proof is analogous to that of the previous theorem.
\end{proof}

Consider now replacing the assumption that $p(a|c)=p(\na|\nc) \geq 0.5$ in Theorem \ref{the:nonmonotone} by the weaker assumption that $p(\na|\nc) \geq p(a|c) \geq 0.5$. Then, $RD_{crude} \geq RD_{obs}$ does not always hold. Our experiments suggest that it holds for approximately 90 \% of the parameterizations.\footnote{Code available at \texttt{https://www.dropbox.com/s/6dy2mr067ukc7zq/nonmonotonicity2.R?dl=0}.} However, $RD_{crude} \geq RD_{true}$ and $RD_{obs} \geq RD_{true}$ always hold, as the following theorem proves. This result is useful when, for instance, $0 > RD_{crude}$ or $0 > RD_{obs}$ because, then, one can readily conclude that $0 > RD_{obs}$, i.e. suffering the disease $A$ reduces the average severity of the disease $Y$.

\begin{theorem}
Consider the causal graph to the left in Figure \ref{fig:graphs}. Let $p(c)=0.5$, $p(d|c)=p(\nd|\nc) \geq 0.5$ and $p(\na|\nc) \geq p(a|c) \geq 0.5$. If $E[Y|a,c] - E[Y|a,\nc] \geq E[Y|\na,\nc] - E[Y|\na,c] \geq 0$, then $RD_{crude} \geq RD_{true}$ and $RD_{obs} \geq RD_{true}$.
\end{theorem}

\begin{proof}
We start by proving that $RD_{crude} \geq RD_{true}$. Recall from the proof of Theorem \ref{the:theorem} that
\[
p(c|a) = \sigma \Big( \ln \frac{p(a|c)p(c)}{p(a|\nc)p(\nc)} \Big) = \sigma \Big( \ln \frac{p(a|c)}{p(a|\nc)} \Big)
\]
where the second equality follows from the assumption that $p(c)=0.5$. Likewise, 
\[
p(\nc|\na) = \sigma \Big( \ln \frac{p(\na|\nc)}{p(\na|c)} \Big).
\]
Note also that $p(c|a) \geq 0.5$ and $p(\nc|\na) \geq 0.5$ due to the assumption that $p(\na|\nc) \geq p(a|c) \geq 0.5$. Now, consider the function $f(x)=x(1-x)$. By inspecting the first and second derivatives, we can conclude that $f(x)$ has a single maximum at $x=0.5$, and that it is increasing in the interval $[0,0.5]$ and decreasing in the interval $[0.5,1]$. This implies that $p(a|c)p(\na|c) \geq p(a|\nc)p(\na|\nc)$ due to the assumption that $p(\na|\nc) \geq p(a|c) \geq 0.5$. Then,
\begin{equation}\label{eq:ffunction}
\frac{p(a|c)}{p(a|\nc)} \geq \frac{p(\na|\nc)}{p(\na|c)}
\end{equation}
which together with the fact that $\sigma()$ and $\ln()$ are increasing functions imply that $p(c|a) \geq p(\nc|\na)$.

The results in the previous paragraph allow us to write $p(c|a)=0.5+\alpha$ and $p(\nc|\na)=0.5+\beta$ with $\alpha \geq \beta \geq 0$. Therefore,
\begin{align*}
RD_{crude} &= E[Y|a] - E[Y|\na]\\
&= E[Y|a,c] p(c|a) + E[Y|a,\nc] p(\nc|a)\\
&- E[Y|\na,c] p(c|\na) - E[Y|\na,\nc] p(\nc|\na)\\
&= E[Y|a,c] (0.5 + \alpha) + E[Y|a,\nc] (0.5 - \alpha)\\
&- E[Y|\na,c] (0.5 - \beta) - E[Y|\na,\nc] (0.5 + \beta)\\
&= RD_{true} + \alpha \big( E[Y|a,c] - E[Y|a,\nc] \big) - \beta \big( E[Y|\na,\nc] - E[Y|\na,c] \big)\\
& \geq RD_{true}
\end{align*}
because $\alpha \geq \beta \geq 0$ as shown above, $E[Y|a,c] - E[Y|a,\nc] \geq E[Y|\na,\nc] - E[Y|\na,c] \geq 0$ by assumption, and
\begin{equation}\label{eq:rdtrue2}
RD_{true} = E[Y|a,c] 0.5 + E[Y|a,\nc] 0.5 - E[Y|\na,c] 0.5 - E[Y|\na,\nc] 0.5.
\end{equation}
due to the assumption that $p(c)=0.5$.

We continue by proving that $RD_{obs} \geq RD_{true}$. First, recall again from the proof of Theorem \ref{the:theorem} that
\[
p(c|a,d) = \sigma \Big( \ln \frac{p(a,d|c)p(c)}{p(a,d|\nc)p(\nc)} \Big) = \sigma \Big( \ln \frac{p(a|c)p(d|c)}{p(a|\nc)p(d|\nc)} \Big)
\]
where the second equality follows from the assumption that $p(c)=0.5$, and the fact that $A$ and $D$ are conditionally independent given $C$ due to the causal graph under consideration. Likewise,
\[
p(\nc|a,\nd) = \sigma \Big( \ln \frac{p(a|\nc)p(\nd|\nc)}{p(a|c)p(\nd|c)} \Big).
\]
Therefore, $p(c|a,d) \geq p(\nc|a,\nd)$ because $\sigma()$ and $\ln()$ are increasing functions and
\[
\frac{p(d|c)}{p(d|\nc)} = \frac{p(\nd|\nc)}{p(\nd|c)}
\]
by the assumption that $p(d|c)=p(\nd|\nc)$, and
\[
\frac{p(a|c)}{p(a|\nc)} \geq \frac{p(a|\nc)}{p(a|c)}
\]
by the assumption that $p(\na|\nc) \geq p(a|c) \geq 0.5$. We can analogously prove that $p(c|a,\nd) \geq p(\nc|a,d)$. Therefore, 
\[
p(c|a,d) p(d) + p(c|a,\nd) p(\nd) \geq p(\nc|a,d) p(d) + p(\nc|a,\nd) p(\nd)
\]
because
\[
p(d)=p(d|c)p(c)+p(d|\nc)p(\nc)=p(d|c)p(c)+p(\nd|c)p(\nc)=0.5
\]
by the assumptions that $p(c)=0.5$ and $p(d|c)=p(\nd|\nc)$. Now, note that
\[
p(\nc|a,d) p(d) + p(\nc|a,\nd) p(\nd) = 1 - ( p(c|a,d) p(d) + p(c|a,\nd) p(\nd) )
\]
which implies that $p(c|a,d) p(d) + p(c|a,\nd) p(\nd) \geq 0.5$. We can analogously prove that $p(\nc|\na,d) p(d) + p(\nc|\na,\nd) p(\nd) \geq 0.5$.

Then, consider the expression
\[
p(\nc|\na,\nd) = \sigma \Big( \ln \frac{p(\na|\nc)p(\nd|\nc)}{p(\na|c)p(\nd|c)} \Big).
\]
Therefore, $p(c|a,d) \geq p(\nc|\na,\nd)$ due to the following three observations. First,
\[
\frac{p(d|c)}{p(d|\nc)} = \frac{p(\nd|\nc)}{p(\nd|c)}
\]
by the assumption that $p(d|c)=p(\nd|\nc)$. Second,
\[
\frac{p(a|c)}{p(a|\nc)} \geq \frac{p(\na|\nc)}{p(\na|c)}
\]
as shown in Equation \ref{eq:ffunction}. Third, $\sigma()$ and $\ln()$ are increasing functions. We can analogously prove that $p(c|a,\nd) \geq p(\nc|\na,d)$. Therefore,
\[
p(c|a,d) p(d) + p(c|a,\nd) p(\nd) \geq p(\nc|\na,d) p(d) + p(\nc|\na,\nd) p(\nd)
\]
because $p(d)=0.5$ as shown above.

Finally, the results in the previous paragraphs allow us to write $p(c|a,d) p(d) + p(c|a,\nd) p(\nd) = 0.5 + \alpha$ and $p(\nc|\na,d) p(d) + p(\nc|\na,\nd) p(\nd) = 0.5 + \beta$ with $\alpha \geq \beta \geq 0$. Consequently, $p(\nc|a,d) p(d) + p(\nc|a,\nd) p(\nd) = 0.5 - \alpha$, and $p(c|\na,d) p(d) + p(c|\na,\nd) p(\nd) = 0.5 - \beta$. Therefore,
\begin{align*}
RD_{obs} &= E[Y|a,d] p(d) + E[Y|a,\nd] p(\nd) - E[Y|\na,d] p(d) - E[Y|\na,\nd] p(\nd)\\
&= \big( E[Y|a,c,d]p(c|a,d) + E[Y|a,\nc,d]p(\nc|a,d) \big) p(d)\\
&+ \big( E[Y|a,c,\nd]p(c|a,\nd) + E[Y|a,\nc,\nd]p(\nc|a,\nd) \big) p(\nd)\\
&- \big( E[Y|\na,c,d]p(c|\na,d) + E[Y|\na,\nc,d]p(\nc|\na,d) \big) p(d)\\
&- \big( E[Y|\na,c,\nd]p(c|\na,\nd) + E[Y|\na,\nc,\nd]p(\nc|\na,\nd) \big) p(\nd)\\
& = E[Y|a,c] (p(c|a,d) p(d) + p(c|a,\nd) p(\nd))\\
& + E[Y|a,\nc] (p(\nc|a,d) p(d) + p(\nc|a,\nd) p(\nd))\\
& - E[Y|\na,c] (p(c|\na,d) p(d) + p(c|\na,\nd) p(\nd))\\
& - E[Y|\na,\nc] (p(\nc|\na,d) p(d) + p(\nc|\na,\nd) p(\nd))\\
& = E[Y|a,c] (0.5 + \alpha) + E[Y|a,\nc] (0.5 - \alpha)\\
& - E[Y|\na,c] (0.5 - \beta) - E[Y|\na,\nc] (0.5 + \beta)
\end{align*}
where the third equality follows from the fact that $Y$ and $D$ are conditionally independent given $A$ and $C$ due to the causal graph under consideration. Then,
\begin{align*}
RD_{obs} &= RD_{true} + \alpha \big( E[Y|a,c] - E[Y|a,\nc] \big) - \beta \big( E[Y|\na,\nc] - E[Y|\na,c] \big)\\
& \geq RD_{true}
\end{align*}
by Equation \ref{eq:rdtrue2}, the above shown fact that $\alpha \geq \beta \geq 0$, and the assumption that $E[Y|a,c] - E[Y|a,\nc] \geq E[Y|\na,\nc] - E[Y|\na,c] \geq 0$.
\end{proof}

One can analogously prove the following result.

\begin{theorem}
Consider the causal graph to the left in Figure \ref{fig:graphs}. Let $p(c)=0.5$, $p(d|c)=p(\nd|\nc) \geq 0.5$ and $p(\na|\nc) \geq p(a|c) \geq 0.5$. If $E[Y|a,c] - E[Y|a,\nc] \leq E[Y|\na,\nc] - E[Y|\na,c] \leq 0$, then $RD_{crude} \leq RD_{true}$ and $RD_{obs} \leq RD_{true}$.
\end{theorem}

\section{On a Driver of the Confounder}\label{sec:driver}

Consider the causal graph to the right in Figure \ref{fig:graphs}. Note that $D$ is now a driver rather than a proxy of $C$, i.e. $D$ causes $C$. The graph entails the following factorization:
\begin{equation}\label{eq:factorization2}
p(A,C,D,Y)=p(D)p(C|D)p(A|C)p(Y|A,C).
\end{equation}
We show next that our previous results also apply to the new causal graph under consideration.

\begin{theorem}\label{the:theorem1}
Consider the causal graph to the right in Figure \ref{fig:graphs}. If $E[Y|A,D]$ is monotone in $D$, then $E[Y|A,C]$ is monotone in $C$.
\end{theorem}

\begin{proof}
The proof of Theorem \ref{the:theorem} also applies when $D$ is a driver of $C$.
\end{proof}

\begin{corollary}\label{cor:corollary2}
Consider the causal graph to the right in Figure \ref{fig:graphs}. If $E[Y|A,D]$ is monotone in $D$, then $RD_{obs}$ lies between $RD_{true}$ and $RD_{crude}$.
\end{corollary}

\begin{proof}
Note that every probability distribution that is representable by the causal graph to the right in Figure \ref{fig:graphs} can be represented by the causal graph to the left in Figure \ref{fig:graphs}: Simply, let $p_L(A|C)=p_R(A|C)$ and $p_L(Y|A,C)=p_R(Y|A,C)$ where the subscript $L$ or $R$ indicates whether we refer to Equation \ref{eq:factorization} or \ref{eq:factorization2}, respectively. Moreover, let
\[
p_L(C) = p_R(C) = p_R(C|d)p_R(d)+p_R(C|\nd)p_R(\nd)
\]
and
\[
p_L(D|C) = p_R(D|C) = \frac{p_R(C|D)p_R(D)}{p_R(C|d)p_R(d)+p_R(C|\nd)p_R(\nd)}.
\]
Therefore, $RD_{crude}$, $RD_{obs}$ and $RD_{true}$ are the same whether they are computed from the graph to the right or to the left in Figure \ref{fig:graphs}. Likewise, if $E[Y|A,D]$ is monotone in $D$ for the graph to the right in Figure \ref{fig:graphs}, then it is also monotone in $D$ for the graph to the left, which implies that $RD_{obs}$ lies between $RD_{true}$ and $RD_{crude}$ by Corollary \ref{cor:corollary}.
\end{proof}

\citet[Result 1]{VanderWeeleetal.2008} prove that (i) if $E[Y|A,C]$ and $E[A|C]$ are both nondecreasing or both nonincreasing in $C$, then $RD_{obs} \geq RD_{true}$, and (ii) if $E[Y|A,C]$ and $E[A|C]$ are one nondecreasing and the other nonincreasing in $C$, then $RD_{obs} \leq RD_{true}$. The antecedents of these rules cannot be verified empirically, because $C$ is unobserved. Therefore, one must rely on substantive knowledge to apply the rules. Luckily, the rules also hold for $D$ and, thus, the antecedents can be verified empirically. The following theorem proves it.

\begin{theorem}\label{the:theorem3}
Consider the causal graph to the right in Figure \ref{fig:graphs}. If $E[Y|A,D]$ and $E[A|D]$ are both nondecreasing or both nonincreasing in $D$, then $E[Y|A,C]$ and $E[A|C]$ are both nondecreasing or both nonincreasing in $C$. If $E[Y|A,D]$ and $E[A|D]$ are one nondecreasing and the other nonincreasing in $D$, then $E[Y|A,C]$ and $E[A|C]$ are one nondecreasing and the other nonincreasing in $C$.
\end{theorem}

\begin{proof}
We prove the result when $E[Y|A,D]$ and $E[A|D]$ are both nondecreasing in $D$. The proofs for the rest of the cases are similar. Then, we have that (i) $E[Y|a, d] \geq E[Y|a, \nd]$, and (ii) $E[Y|d] \geq E[Y|\nd]$. Assume to the contrary that (iii) $E[Y|a, c] \leq E[Y|a, \nc]$, and (iv) $E[Y|c] \geq E[Y|\nc]$. As shown in the proof of Theorem \ref{the:theorem}, (i) and (iii) imply that $p(c|a,d) \leq p(c|a, \nd)$, which implies that
\[
\frac{p(d|c)}{p(d|\nc)} \leq \frac{p(\nd|c)}{p(\nd|\nc)}.
\]
Likewise, (ii) and (iv) imply that $p(c|d) \geq p(c|\nd)$, which implies that
\[
\frac{p(d|c)}{p(d|\nc)} \geq \frac{p(\nd|c)}{p(\nd|\nc)}.
\]
As shown in the proof of Theorem \ref{the:theorem}, this contradicts the fact that $C$ and $D$ are dependent. Therefore, either the assumption (iii) or (iv) or both are false. In the latter case, we get a similar contradiction. So, either the assumption (iii) or (iv) is false. We reach a similar contradiction if replace $a$ with $\na$ in the assumptions (i) and (iii). This together with the fact that $E[Y|A,C]$ and $E[A|C]$ are both monotone in $C$ by Theorem \ref{the:theorem} prove the result.
\end{proof}

\begin{corollary}\label{cor:corollary3}
Consider the causal graph to the right in Figure \ref{fig:graphs}. If $E[Y|A,D]$ and $E[A|D]$ are both nondecreasing or both nonincreasing in $D$, then $RD_{obs} \geq RD_{true}$. If $E[Y|A,D]$ and $E[A|D]$ are one nondecreasing and the other nonincreasing in $D$, then $RD_{obs} \leq RD_{true}$.
\end{corollary}

\begin{proof}
The result follows directly from Theorem \ref{the:theorem3} and \citet[Result 1]{VanderWeeleetal.2008}.
\end{proof}

For completeness, we show below that the converse of Theorem \ref{the:theorem3} also holds.

\begin{theorem}\label{the:theorem4}
Consider the causal graph to the right in Figure \ref{fig:graphs}. If $E[Y|A,C]$ and $E[A|C]$ are both nondecreasing or both nonincreasing in $C$, then $E[Y|A,D]$ and $E[A|D]$ are both nondecreasing or both nonincreasing in $D$. If $E[Y|A,C]$ and $E[A|C]$ are one nondecreasing and the other nonincreasing in $C$, then $E[Y|A,D]$ and $E[A|D]$ are one nondecreasing and the other nonincreasing in $D$.
\end{theorem}

\begin{proof}
As shown in the proof of Corollary \ref{cor:corollary2}, every probability distribution that is representable by the causal graph to the right in Figure \ref{fig:graphs} can be represented by the causal graph to the left in Figure \ref{fig:graphs}. Therefore, if $E[Y|A,C]$ and $E[A|C]$ are monotone in $C$ for the right graph, then they are so for the left graph as well. Then, $E[Y|A,D]$ and $E[A|D]$ are monotone in $D$ for the left graph \cite[Lemma 1]{OgburnandVanderWeele2012a} and, thus, they are so for the right graph as well. The result follows now from the contrapositive formulation of Theorem \ref{the:theorem3}.
\end{proof}

Given a sufficiently large sample from $p(A,D,Y)$, we may conclude from it that $E[Y|A,D]$ is monotone in $D$, which implies that $RD_{obs}$ lies between $RD_{true}$ and $RD_{crude}$ by Corollary \ref{cor:corollary2}. We can also estimate $RD_{obs}$ and $RD_{crude}$ from the sample, which implies that (i) if $RD_{crude} \leq RD_{obs}$ then $RD_{obs} \leq RD_{true}$, and (ii) if $RD_{crude} \geq RD_{obs}$ then $RD_{obs} \geq RD_{true}$. Consequently, Corollary \ref{cor:corollary3} is superfluous when data over $(A,D,Y)$ are available. The following example illustrates that the corollary may be useful when no such data are available.

\begin{example}
Let $A$ and $Y$ represent a treatment and a disease, respectively. Let $D$ and $C$ represent pre-treatment covariates such as socio-economic and health status, respectively. Say that we have a sample from $p_1(A,D,Y)$ and a sample from $p_2(A,D,Y)$, i.e. we have two samples from two different populations. We are interested in drawing conclusions about $RD_{true}$ for a third population, from which we have no data. We make the following assumptions:
\begin{itemize}
\item $p_1(D) \neq p_3(D) \neq p_2(D)$ because the socio-economic profile of the third population differs from the other populations' profiles. 
\item $p_1(C|D) = p_2(C|D) = p_3(C|D)$ because this distribution represents psychological and physiological processes shared by the three populations.
\item $p_1(Y|A,C)=p_3(Y|A,C) \neq p_2(Y|A,C)$ because these distributions represent psychological and physiological processes shared by the first and third populations but not by the second. Then, $E_3[Y|A,D] = E_1[Y|A,D]$ which can be estimated from the sample from $p_1(A,D,Y)$.
\item $p_1(A|C) \neq p_2(A|C) = p_3(A|C)$ because the second and third populations share the same treatment policy but the first does not. Then, $E_3[A|D] = E_2[A|D]$ which can be estimated from the sample from $p_2(A,D,Y)$. 
\end{itemize}
Then, we cannot estimate $RD_{crude}$ for the third population and, thus, we cannot use Corollary \ref{cor:corollary2} as we did before to bound $RD_{true}$. Corollary \ref{cor:corollary3} may, on the other hand, be useful in drawing conclusions. For instance, assume that $E_3[Y|A,D]$ and $E_3[A|D]$ are both nondecreasing or both nonincreasing in $D$. Then, $RD_{obs} \geq RD_{true}$ by the corollary. If we are interested in testing whether $k \geq RD_{true}$ for a given constant $k$, then it may be worth assuming the cost of collecting data from the third population in order to compute $RD_{obs}$, in the hope that $k \geq RD_{obs}$ which confirms the hypothesis. If we are interested in testing whether $RD_{true} \geq k$, then we may also be willing to assume the cost, in the hope that $k \geq RD_{obs}$ which allows us to reject the hypothesis. In the latter case, we may instead decide to not assume the cost because we can never confirm the hypothesis. Such a seemingly negative result may save us time and money. Similar conclusions can be drawn when $E[Y|A,D]$ and $E[A|D]$ are one nondecreasing and the other nonincreasing in $D$. On the other hand, no such conclusions can be drawn from Corollary \ref{cor:corollary2} before collecting data.
\end{example}

\subsection{Bounds}

Causal effects are typically defined in terms of distributions of counterfactuals. For instance, the causal effect on $Y$ of an intervention setting $A=a$ is defined as $E[Y_a]$. It can be rewritten as follows \cite[Theorem 3.3.2]{Pearl2009}:
\[
E[Y_a] = E[Y|a,c]p(c) + E[Y|a,\nc]p(\nc).
\]
Since $C$ is unobserved, this effect cannot be computed. It can be approximated by the following quantity:
\[
S_a = E[Y|a,d]p(d) + E[Y|a,\nd]p(\nd).
\]
It can also be approximated by $S_a = E[Y|a]$. Likewise for the causal effect on $Y$ of an intervention setting $A=\na$. The following discussion applies to both approximations.

\citet[Result 1]{VanderWeeleetal.2008} prove that (i) if $E[Y|A,C]$ and $E[A|C]$ are both nondecreasing or both nonincreasing in $C$, then $S_a \geq E[Y_a]$ and $S_{\na} \leq E[Y_{\na}]$, and (ii) if $E[Y|A,C]$ and $E[A|C]$ are one nondecreasing and the other nonincreasing in $C$, then $S_a \leq E[Y_a]$ and $S_{\na} \geq E[Y_{\na}]$. These results also hold when $E[Y|A,D]$ and $E[A|D]$ are nondecreasing or nonincreasing in $D$ by Theorem \ref{the:theorem3}. The following corollary shows that the results also hold under weaker assumptions: It is not necessary that $E[Y|A,D]$ is nondecreasing or nonincreasing in $D$, it suffices with $E[Y|a,D]$ and $E[Y|\na,D]$ being so, which is always true. Specifically, we say that $E[Y|a,D]$ is nondecreasing in $D$ if 
\[
E[Y|a, d] \geq E[Y|a, \nd]
\]
and we say that it is nonincreasing in $D$ if
\[
E[Y|a, d] \leq E[Y|a, \nd].
\]
Likewise for $E[Y|\na, D]$.

\begin{corollary}
Consider the causal graph to the right in Figure \ref{fig:graphs}. If $E[Y|a,D]$ and $E[A|D]$ are both nondecreasing or both nonincreasing in $D$, then $S_a \geq E[Y_a]$. If $E[Y|a,D]$ and $E[A|D]$ are one nondecreasing and the other nonincreasing in $D$, then $S_a \leq E[Y_a]$. Likewise for $\na$ instead of $a$ replacing $\geq$ with $\leq$ and vice versa.
\end{corollary}

\begin{proof}
We prove the result for when $E[Y|a,D]$ and $E[A|D]$ are both nondecreasing in $D$. The proof is similar for the remaining cases. If $E[Y|\na,D]$ is not nondecreasing in $D$, then make it so by parameterizing $p(Y|\na,C)$ appropriately in Equation \ref{eq:factorization2}, e.g. by setting $p(Y|\na,c) = p(Y|\na,\nc)$ so that $E[Y|\na,d] = E[Y|\na,\nd]$. Then, as discussed previously, $S_a \geq E[Y_a]$ for the new distribution. Finally, note that the expressions for $S_a$ and $E[Y_a]$ do not involve $p(Y|\na,C)$. So, $S_a$ and $E[Y_a]$ are the same for the new and the original distributions.
\end{proof}

Of course, $S_a$ is always an upper or lower bound of $E[Y_a]$. The previous corollary allows us to determine always whether it is the one or the other, because $E[Y|a,D]$ and $E[A|D]$ are always nondecreasing or nonincreasing in $D$. Likewise for $\na$ instead of $a$. On the other hand, given a random parameterization, there is only 50 \% chance that $E[Y|a,D]$ and $E[Y|\na,D]$ are both nondecreasing or both nonincreasing in $D$ and, thus, $E[Y|A,D]$ is nondecreasing or nonincreasing in $D$ and, thus, we can apply the combination of Theorem \ref{the:theorem3} and the result by \citet[Result 1]{VanderWeeleetal.2008} as we did above.

\subsection{Transitivity}

Consider the causal graph $A \ra B \ra C$. Let $E[B|A]$ and $E[C|B]$ be nondecreasing in $A$ and $B$, respectively. Unfortunately, there is no guarantee that $E[C|A]$ is nondecreasing in $A$, i.e. the nondecreasing property is not transitive in general \cite[Example 3.2]{VanderWeeleandRobins2010}. However, transitivity does hold when $A$, $B$ and $C$ are binary random variables \cite[p. 119]{VanderWeeleandRobins2010}. For binary random variables, \citet[Lemma 1]{OgburnandVanderWeele2012a} also implies a sort of transitivity result: If $E[C|B]$ is monotone in $B$, then $E[C|A]$ is monotone in $A$. Theorem \ref{the:theorem1} implies then a sort of inverse transitivity result: If $E[C|A]$ is monotone in $A$, then $E[C|B]$ is monotone in $B$.

\section{Discussion}\label{sec:discussion}

We have extended the result in Lemma 1 by \citet{OgburnandVanderWeele2012a} stating that if $E[Y|A,C]$ is monotone in $C$, then $RD_{obs}$ lies between $RD_{true}$ and $RD_{crude}$. We have done so by showing that the result also holds when $E[Y|A,D]$ is monotone in $D$. This makes the result much more applicable in practice, as the monotonicity condition in $D$ can be verified empirically. We have also extended along the same lines the results reported in Result 1 by \citet{VanderWeeleetal.2008}. 

The monotonicity condition in $D$ is, however, sufficient but not necessary. In fact, we have shown through experiments that 94 \% of the random parameterizations of the causal graph studied resulted in $RD_{obs}$ being inside the range of $RD_{true}$ and $RD_{crude}$. However, the monotonocity condition did not hold for approximately half of them. To shed some light on this question, we have characterized a nonmonotonic case (albeit empirically untestable) where $RD_{obs}$ still lies between $RD_{true}$ and $RD_{crude}$. In future work, we plan to investigate how to relax the monotonicity condition while keeping it sufficient and empirically testable.

\section*{Acknowledgments}

We thank the Associate Editor and Reviewers for their comments, which helped us to improve our work. This work was funded by the Swedish Research Council (ref. 2019-00245).

\bibliographystyle{plainnat}
\bibliography{monotonicity}

\end{document}